\documentclass[a4paper,11pt]{extarticle}

\usepackage{fullpage,amsmath,amsthm,amssymb,url}

\usepackage[linesnumbered,vlined,ruled]{algorithm2e}
\usepackage{mathscinet} 

\newtheorem{theorem}{Theorem}
\newtheorem{lemma}[theorem]{Lemma}
\newtheorem{proposition}[theorem]{Proposition}
\newtheorem*{thmone}{Theorem \ref{thm:22}}
\newtheorem*{thmtwo}{Theorem \ref{thm:coprime}}

\theoremstyle{definition}
\newtheorem{question}{Question}

\theoremstyle{remark}
\newtheorem{case}{Case}

\newcommand{\ie}{i.e.\ }
\newcommand{\hcf}{\operatorname{hcf}}

\SetKw{see}{See}
\SetKw{dec}{Declare}
\SetKw{flip}{Flip}
\SetKw{rpttext}{Repeat}
\SetKw{wotext}{Without}
\SetKw{eltext}{Else}
\SetKw{pass}{Pass}
\SetKwProg{rpt}{Repeat}{}{}
\SetKwIF{If}{ElseIf}{Else}{Without}{}{}{Else}{}

\usepackage{authblk}
\title{A parallel wakeup problem and multi-room light switch strategies}
\author[1]{John Haslegrave}
\author[2]{Paul A. Russell}
\author[3]{Mark Walters}
\affil[1]{Lancaster University}
\affil[2]{Churchill College, Cambridge}
\affil[3]{Queen Mary University of London}
\begin{document}
	\maketitle
	
	\begin{abstract}
		The wakeup problem in distributed computing asks for a symmetric protocol that enables one of several processors to eventually guarantee that all (or, in a more general setting, enough) other processors have acted, using a shared register but no global clock. Dropping the symmetry requirement gives a well-known exercise often phrased in terms of prisoners entering, in an unknown sequence, a room equipped with a single binary switch, and using it to communicate. Kane and Kominers recently analysed a more general version of the latter with multiple parallel and indistinguishable rooms. We answer some open questions of Kane and Kominers regarding the minimum number of switch states needed for the prisoners to solve the problem. We also consider the symmetric ``wakeup'' version of this scenario, and establish exactly for which numbers of processors and registers a solution is possible.
	\end{abstract}
	
	\section{Introduction}
	The wakeup problem in distributed computing was introduced by Fischer, Moran, Rudich and Taubenfeld \cite{FMRT90,FMRT96}, motivated by the fact that previous assumptions of processes acting synchronously and having unique names known to one another are not realistic. It requires a protocol for $n$ identical processes acting completely asynchronously in a shared-memory environment, such that eventually at least $p$ processes will know that at least $\tau$ processes have ``woken up'', i.e.\ begun the protocol, and this is resilient to the failure of up to $t$ processes. In the simplest case, which we restrict our attention to here, we set $p=1$, $\tau=n$ and $t=0$, i.e.\ processes are fully reliable, and at least one process must know that all of them have woken up. They show that this can be done with a shared memory of just two bits. Variants of the wakeup problem have since been studied in a variety of contexts; see e.g.\ \cite{AGM02,JS05,CGK07}.
	
	The requirement that processes act identically arises naturally in self-stabilisation problems, introduced earlier by Dijkstra \cite{Dij74}, where there may be no natural ``leader'' and the agents need not know each other's identity, meaning that they cannot initially agree on any sort of priority between themselves. Indeed, establishing priority may be precisely the goal of self-stabilisation, as in Herman's algorithm \cite{Her90}.
	Dropping this requirement, while retaining the asynchronicity, enables the wakeup problem to be solved with only one bit of shared memory, using a predetermined leader that will eventually learn that all others have woken up. This observation appears to be the basis for a problem with a more down-to-earth setting that began circulating some years after the introduction of the wakeup problem. It appeared with solution on IBM's website \textit{Ponder This} \cite{ponder}, subsequently was featured in Berlekamp and Buhler's column in \textit{The Emissary} \cite{BB02}, and was brought to the attention of a wider audience in Winkler's book \cite{Win04}.
	
	\begin{quote}A warden guards $n$ prisoners. There is a room with a single light switch; the warden leads prisoners into the room in some sequence, with the only restriction being that every prisoner must eventually visit the room arbitrarily many times. On visiting the room, a prisoner may observe the current state of the switch, and choose whether or not to change it. Prisoners do not have any other information; in particular, they cannot deduce anything about the number of visits that have occurred based on the passing of time. They win if one of them can eventually declare with certainty that all prisoners have visited the room.\end{quote} 
	
	It is not hard to show that the prisoners have a winning strategy for any value of $n$, whether or not the initial state of the switch is known. Winkler also discusses the question of symmetric strategies, in which all prisoners are required to follow the same protocol, relating this to the wakeup problem.
	
	Kane and Kominers \cite{KK21} recently analysed a generalisation of this problem to a setting with multiple, indistinguishable rooms, although with Chen they introduced a form of this problem (where the number of states can only be a power of $2$) as a puzzle much earlier \cite{KKC08}.
	
	In the generalised version considered in \cite{KK21}, there are $r>1$ indistinguishable 
	rooms each having $q$ possible states.
	The warden leads prisoners into rooms in some sequence, with the restriction that each prisoner must eventually visit each room arbitrarily many times. The initial state of the rooms is known (we will frequently assume that every room is initially in state $0$). The prisoners win if one of them can eventually deduce that all prisoners have visited all rooms. The reasons for specifying the problem in this way are that if the prisoners only needed to deduce that every prisoner has visited a room, they could easily adapt a winning strategy for the one-room case, whereas if the initial state is unknown then no winning strategy is possible \cite{KK21}.
	
	If a winning strategy uses randomness, we may determine the results of all possible random choices in advance to give a deterministic winning strategy. Thus we assume throughout that all strategies are deterministic.
	
	Kane and Kominers \cite{KK21} established the following results, for $n\geq 2$ prisoners in $r$ rooms with $q$ states.
	\begin{itemize}
		\item The prisoners always have a winning strategy if $q\geq 4$.
		\item The prisoners do not have a winning strategy if $q=2$ and $r\geq 3$.
		\item The prisoners have a winning strategy if $q=3$ and $r=2$.
		\item The prisoners have a winning strategy if $q=3$ and $n=2$.
	\end{itemize}
	The following were given as open questions.
	\begin{question}[Remark 11 of \cite{KK21}]\label{qu:22}Who wins in the case $r=q=2$?\end{question}
	\begin{question}[Section 4 of \cite{KK21}]\label{qu:r3}Do the prisoners have a winning strategy if $q=3$?\end{question}
	\begin{question}[Section 6.4 of \cite{KK21}]\label{qu:sym}What happens if the prisoners must use a symmetric
		strategy?\end{question}
	
	Our first result, proved in Sections \ref{sec:mono} and \ref{sec:22}, gives a complete answer to Question \ref{qu:22}.
	\begin{theorem}\label{thm:22}
		The prisoners do not have a winning strategy for $r=q=2$ (for any $n\geq 2$).
	\end{theorem}
	
	We also make progress on Question \ref{qu:r3} by showing in Section \ref{sec:three} that four or fewer prisoners can always win, and that five prisoners can win unless the number of rooms is exactly three.
	
	Question \ref{qu:sym} can be thought of as a parallel version of the wakeup problem, using $r$ indistinguishable shared memory registers with the requirement that all processes have ``woken up'' with respect to all $r$ separately. However, for ease of reading, we will continue to use the terminology of light switches when addressing this question in Section \ref{sec:symm}. As previously noted, the requirement to use symmetric strategies is natural in self-stabilisation problems. In this paper we seek a strategy guaranteed to succeed even against an adversary who tries to counter it, which is the strongest form of self-stabilisation. However, often a strategy that works almost surely against a random adversary (``probabilistic stabilisation'') is acceptable in applications; see \cite{DTY15} for a discussion of the relationship between these types of stabilisation. 
	
	When restricted to symmetric strategies, it is not always possible for the prisoners to win even if the number of states is large. Our final result determines precisely the pairs $(n,r)$ such that the prisoners have a winning strategy for some $q$.
	\begin{theorem}\label{thm:coprime}
		If there are $n$ prisoners and $r$ rooms, the prisoners have a symmetric winning strategy for some $q$ if and only if $\hcf(n,r)=1$. Furthermore, there is some function $f$ such that for every pair $n,r$ with $\hcf(n,r)=1$, the prisoners have a winning strategy, starting from all rooms in state $0$, for any $q\geq f(\min(n,r))$.
	\end{theorem}
	
	\section{Monotonicity}\label{sec:mono}
	A difficulty with these problems is the lack of straightforward arguments for monotonicity. Of course, there is monotonicity in $q$, since if the prisoners can win for a particular $q$ they can also win for $q'>q$ simply by never using the last $q'-q$ states (if some rooms are initially in these states, then the prisoners may adopt a winning strategy for the initial state where these are replaced by state $0$, with the addendum that any prisoner who sees an extraneous state changes it to state $0$ forthwith). One might imagine (at least for non-symmetric strategies) that having more prisoners can only make things harder, and in particular that if $n$ prisoners $P_1,\ldots, P_n$ have a winning strategy (for particular values of $r$ and $q$) then we may obtain a winning strategy for $n-1$ prisoners $P'_1,\ldots,P'_{n-1}$ by $P'_i$ following the strategy of $P_i$ for each $i<n-1$, and $P'_{n-1}$ choosing at each turn whether to act as $P_{n-1}$ or $P_{n}$. Unfortunately, this argument breaks down since this strategy is only guaranteed to win if $P'_{n-1}$ enters each room infinitely often as $P_{n-1}$ and infinitely often as $P_n$, which we cannot ensure. (Recall that the rooms are indistinguishable, so $P'_{n-1}$ cannot adopt a simple strategy of, say, for each room acting alternately as $P_{n-1}$ and as $P_n$.) However, a variation of this argument can be used to show monotonicity.
	
	\begin{lemma}\label{monotonic}Suppose that $n$ prisoners do not have a winning strategy for $r$ rooms with $q$ states. Then $n+1$ prisoners do not have a winning strategy for $r$ rooms with $q$ states.
	\end{lemma}
	\begin{proof}Clearly we have $n>1$. Fix an arbitrary strategy for each of $n+1$ prisoners $P_1,\ldots, P_{n+1}$. We will show that the warden has a strategy that defeats it. To this end, consider $n$ prisoners $P'_1,\ldots, P'_n$ with the following strategy. For each $i<n$, $P'_i$ behaves exactly as $P_i$ does in the original strategy. Prisoner $P'_n$, on entering a room in state $s$, first behaves as prisoner $P_n$ would on entering a room in state $s$, updating the state to $s'$, then behaves as prisoner $P_{n+1}$ would on entering a room in state $s'$.
		
		By assumption, the warden has a counter-strategy against prisoners $P'_1,\ldots,P'_n$, that is, a schedule for leading the prisoners into rooms such that none of them can ever deduce that all prisoners have visited all rooms, and yet each prisoner eventually visits every room arbitrarily many times. The warden may adapt this schedule to prisoners $P_1,\ldots, P_n$ by replacing each visit of $P'_i$ for $i<n$ with a visit of $P_i$ to the same room, and replacing each visit of $P'_n$ by consecutive visits of $P_n$ and $P_{n+1}$ to the same room. This ensures that each of $P_1,\ldots, P_{n+1}$ eventually visits every room arbitrarily many times. Note that if all prisoners have visited all rooms by a certain point in this schedule, then the same is true for the corresponding point in the original schedule. Since, for each $i<n$, $P_i$ has the same information that $P'_i$ would have in the original schedule, none of $P_1,\ldots,P_{n-1}$ can ever deduce that all prisoners have visited all rooms. Similarly, since $P_n$ and $P_{n+1}$ each have no more information about $P_1,\ldots, P_{n-1}$ than $P'_n$ would have had about $P'_1,\ldots,P'_{n-1}$, neither of them can ever deduce that all of $P_1,\ldots,P_{n-1}$ have visited all rooms. Thus the strategy we chose is not a winning strategy for the prisoners; since it was arbitrary, they do not have a winning strategy.
	\end{proof}
	
	The final parameter is the number of rooms, $r$. There is no apparent reason why we should expect monotonicity in $r$. Indeed, when $n=5$ and $q=3$ we know that the prisoners win if $r<3$ or if $r>3$, but the problem is open for $r=3$.
	
	We recall that if the prisoners are required to use a symmetric strategy, there is no monotonicity in either $n$ or $r$: see Theorem \ref{thm:coprime}.
	
	\section{Two rooms with two states}\label{sec:22}
	In this section we prove the following main result, which answers Question \ref{qu:22}.
	
	\begin{thmone}
		The prisoners do not have a winning strategy for $r=q=2$ (for any $n\geq 2$).
	\end{thmone}
	\begin{proof}We show this for two prisoners, Alice and Bob. The result then follows by induction and Lemma \ref{monotonic}.
		
		Assume that a winning strategy does exist. Now consider what happens if Alice, following her strategy, is repeatedly led into a single room initially in state $0$. This defines an infinite sequence of states, which either is eventually constant or takes both values infinitely often. If it becomes constant with value $x$, we sat that Alice is \textit{stuck} in state $x$ from that point, \ie she will not change the state (or declare) until she sees a different state. Similarly, we can define an infinite sequence for Bob with the same possibilities.
		
		Suppose that for some state $x$, both sequences take value $x$ infinitely often. In this case the warden can construct a sequence as follows. Choose $a_1,a_2\ldots$ to be a $1$-$2$ sequence taking each value infinitely often. Now we construct a schedule in rounds. In round $i$, repeatedly lead Alice into room $a_i$ until the next occasion that it is in state $x$, then lead Bob into room $3-a_i$ until the next occasion it is in state $x$. Note that whenever a prisoner enters a room, its state matches that in which they left the previous room. Consequently neither prisoner can distinguish any schedule of this form from any other, and in particular they cannot tell whether $a_1,\ldots, a_k$ includes both values for any $k$, so cannot declare.
		
		Consequently one prisoner (without loss of generality Alice) gets stuck in state $0$ and the other in state $1$. Consider now the following hypothetical procedure. (Note that this procedure as it stands is not something that the warden would be allowed to carry out as part of a strategy, but instead it will be used in the construction of a valid strategy for the warden.) Lead Alice into a room until she gets stuck in $0$, then reset the room state to $1$ and lead Alice in until she gets stuck in $0$ (if ever), then repeat. Either for some $k\geq 1$, after resetting the state to $1$ on $k\geq 1$ occasions Alice will return the room to state $1$ infinitely often, or after every reset she will eventually get stuck in $0$; we take $k=\infty$ in the latter case.
		
		Similarly, we can do the same with Bob, instead resetting the room to $0$; either there is some $\ell\geq 1$ such that after $\ell$ resets Bob will return the room to state $0$ infinitely often, or after every reset he will eventually get stuck in $1$, and we take $\ell=\infty$ in the latter case.
		
		We distinguish four cases depending on the values of $k$ and $\ell$. In each case we choose a sequence $a_1,a_2\ldots$ as before.
		
		\begin{case}$k=\ell=\infty$.\end{case}
		In round $1$, repeatedly lead Alice into room $a_1$ until she becomes stuck in state $0$. For each $i>1$, in round $i$, lead Bob into room $a_i$ until he becomes stuck in state $1$, then lead Alice into room $a_i$ until she becomes stuck in state $0$. Note that neither prisoner can tell whether $a_1,\ldots, a_k$ includes both values for any $k$, so they can never declare. 
		
		\begin{case}$k>\ell$.\end{case}
		In round $1$, lead Alice into room $a_1$ until she is stuck in state $0$. For $1<i\leq\ell$, in round $i$, lead Bob into room $a_i$ until he becomes stuck in state $1$, then Alice into room $a_i$ until she becomes stuck in state $0$. In round $\ell+1$, lead Alice into room $a_{\ell+1}$ until she becomes stuck in state $0$ and then lead Bob into room $a_{\ell+1}$ until he next leaves it in state $0$. In round $i$ for $i>\ell+1$, lead Alice into room $a_i$ once (since she is stuck at $0$, she will not change the state), and then lead Bob into room $a_i$ until he next leaves it in state $0$. Again, neither prisoner can tell whether $a_1,\ldots, a_k$ includes both values for any $k$, so they can never declare. 
		
		\begin{case}$k<\ell$.\end{case}
		In round $1$, lead Alice into room $a_1$ until she is stuck in state $0$. For $1<i\leq k$, in round $i$, lead Bob into room $a_i$ until he becomes stuck in state $1$, then Alice into room $a_i$ until she becomes stuck in state $0$. In round $k+1$, lead Bob into either room until he is stuck in state $1$, then lead him into the other room until he is stuck in state $1$; since Bob did nothing in round $1$, and $\ell\geq k+1$, this is possible. After this round, both rooms are in state $1$ and Bob has entered both rooms but Alice may not have done. For $i>k+1$, in round $i$ lead Alice into room $a_i$ until she next leaves it in state $1$, then lead Bob into room $a_i$ once (since he is stuck in state $1$, he will not change its state). At no point will either prisoner be able to determine whether Alice has visited both rooms.
		
		\begin{case}$k=\ell<\infty$.\end{case}
		In round $1$, lead Alice into room $1$ until she is stuck in state $0$. For $1<i\leq k$, in round $i$, lead Bob into room $1$ until he becomes stuck in state $1$, then Alice into room $1$ until she becomes stuck in state $0$. In round $k+1$, lead Bob into room $1$ until he becomes stuck in state $1$. After $k+1$ rounds, the rooms are in different states, and neither prisoner has visited room $2$. 
		Note that if we lead either prisoner into a room in a given state $s$, and thereafter ensure that each room they visit is in the state in which they left the previous room, they will return the room to state $s$ infinitely often. For Alice with $s=0$, or Bob with $s=1$, this is true because they are stuck in $s$, and otherwise it is true because $\ell=k$, so if led repeatedly into a room initially in state $s$, they will not get stuck in $1-s$ again.
		
		First, suppose that from this point on there is a choice of prisoner $P$, a state $s$, and a sequence of rooms $a_1,a_2,\ldots$, taking both values infinitely often, such that leading prisoner $P$ into rooms $a_1,a_2,\ldots$ in sequence will result in both rooms being simultaneously in state $s$ infinitely many times. Then the warden may follow this sequence for $P$, but interleaved at points where both rooms are in state $s$ by leading the other prisoner $P'$ repeatedly into one room until they return it to state $s$. Since either room can be chosen each time this is done, with no way for either prisoner to tell which room was chosen, this can be done in such a way as to create a valid sequence, and neither prisoner can ever be sure that $P'$ has visited room $2$.
		
		Secondly, suppose that this is not the case. Then leading either prisoner into rooms in any sequence that uses both rooms infinitely often will ensure that this prisoner eventually takes no further action, and the rooms are in opposite states at this point. Choose such sequences $a_1,a_2,\ldots$ and $b_1,b_2,\ldots$. Let $t_A$ be chosen such that if Alice is led into rooms following sequence $a_1,a_2,\ldots$ then her last action will be earlier than $t_A$. Similarly choose $t_B$ for Bob and sequence $b_1,b_2,\ldots$; note that by definition we can choose $t_A,t_B$ to be arbitrarily large. We say that a prisoner ``flips'' if after following their sequence until their last action, room $1$ is in state $0$ (that is, this process has exchanged the states of the rooms). Define $b'_i$ to be $3-b_i$ if Alice flips, and $b_i$ otherwise, and similarly for $a'_i$. Now the warden will adopt either the following strategy, or the equivalent strategy with the roles of Alice and Bob reversed. First lead Alice into rooms $a_1,\ldots,a_{t_A}$, then lead Bob into rooms $b'_1,\ldots,b'_{t_B}$. Subsequently, for each $i\geq 1$ lead Alice into room $a'_{t_A+i}$ followed by Bob into $b'_{t_B+i}$. Note that neither prisoner knows whether Alice or Bob goes first, and since $t_A$ or $t_B$ are unbounded neither prisoner can ever tell that the other has entered room $2$.
	\end{proof}
	
	\section{Winning strategies for three states}\label{sec:three}
	In this section we show that if rooms have three states then four prisoners have a winning strategy for any number of rooms, and five prisoners have a winning strategy unless there are exactly three rooms. We refer to the prisoners as Alice, Bob, Charles, Deborah and Eve.
	
	We describe prisoners' strategies mostly as in \cite{KK21}, but will need an additional control-flow structure. We begin with a brief self-contained outline of the notation, including our new instructions.
	\begin{itemize}
		\item \flip{$(a,b)$} is an instruction to wait until you enter a room in state $a$, then change it to $b$; we write \flip{$(\bullet,b)$} if the prior state of the room can be anything other than $b$.
		\item \see{$(a)$} is an instruction to wait until you enter a room in state $a$, but not change it; it guarantees that no further instructions will be executed until that state has been seen. 
		\item \rpttext{$(k)$} is an instruction to carry out the following indented block $k$ times. 
		\item \dec is an instruction to claim that all prisoners have entered all rooms. 
		\item We add another structure, \wotext{$\cdots$}\eltext. \wotext{$(a)$} is an instruction to carry out the following indented block until either it is completed or you enter a room in state $a$; in the latter case, immediately skip ahead to the following \eltext block. Note that this check is carried out every time you enter a room, which may occur multiple times while processing a single instruction, but is not carried out when leaving a room. The \eltext block is not executed if the \wotext block is completed successfully. 
	\end{itemize}
	We give a simple example to illustrate the new structure.
	
	\NoCaptionOfAlgo
	\DontPrintSemicolon
	\begin{algorithm}[H]
		\caption{Example}
		\eIf{$(0)$}{\flip{$(1,0)$}}{\flip{$(1,2)$}}
	\end{algorithm}
	This fragment requires the prisoner to change the next room encountered in state $1$ to either state $0$ or state $2$ depending on whether a room in state $0$ was seen in the interim. In either case only one flip is carried out.
	
	\begin{theorem}\label{three-four}The prisoners have a winning strategy whenever $n\leq 4$ and $q\geq 3$.
	\end{theorem}
	\begin{proof}
		We may assume that $r\geq 3$, since the two-room case is covered in \cite{KK21} (and the one-room case is a special case of the folklore result that the prisoners win for any $n,q\geq 2$). We start by giving a winning strategy for four prisoners, since, by Lemma \ref{monotonic}, this will prove the full result. We will show that the following strategy works on the assumption that all rooms are initially in state $0$. 
		
		\medskip
		\begin{algorithm}[H]
			\caption{Alice's Protocol}
			\rpt{$(r-1)$}{\flip$(0,1)$}
			\flip{$(0,2)$}\;
		\end{algorithm}
		
		\begin{algorithm}[H]
			\caption{Bob's Protocol}
			\see{$(2)$}\;
			\rpt{$(r-2)$}{\flip{$(1,2)$}}
			\flip{$(1,0)$}
		\end{algorithm}
		
		\begin{algorithm}[H]
			\caption{Charles' Protocol}
			\see{$(2)$}\;
			\see{$(0)$}\;
			\rpt{$(r-1)$}{\flip{$(2,0)$}}
			\flip{$(0,2)$}\;
			\flip{$(0,1)$}
		\end{algorithm}
		
		\begin{algorithm}[H]
			\caption{Deborah's Protocol}
			\see{$(2)$}\;
			\see{$(0)$}\;
			\see{$(1)$}\;
			\rpt{$(r-1)$}{\flip{$(\bullet,1)$}}
			\dec
		\end{algorithm}
		
		Note that the assumptions of \cite[Lemma 4]{KK21} are satisfied with Alice being the leader, and so the prisoners can adapt this strategy to win from any given initial state. We now verify correctness of this strategy.
		
		Since Bob, Charles and Deborah all start with an instruction to see state $2$, only Alice will be able to act initially. Since all rooms start in state $0$, and Alice eventually visits all rooms, she will be able to flip $0$ to $1$ the $r-1$ times required to complete the \rpttext block, and finally flip $0$ to $2$. No other prisoner will act during this time, since there are no rooms in state $2$ until Alice has completed her instructions. Once she has, this signals to Alice that she has visited all rooms and should not act again.
		
		Now the rooms are in states $1,\ldots,1,2$. In particular, all rooms are in state $1$ or $2$.   Whether or not Charles and Deborah fulfil line 1 of their protocols at this time, they are unable to fulfil line 2 until some room returns to state $0$, which will not occur until Bob finishes his actions. So only Bob is able to act; he will eventually see a room in state $2$, and visit all other rooms, flipping $r-2$ of them to state $2$ and the last to state $0$. Once a room returns to state $0$, Bob must have visited all rooms. Now the rooms are in states $0,2,\ldots,2$, and Bob will not act again.
		
		Since there are now rooms in states $0$ and $2$ but not $1$, Deborah may reach line 3 of her protocol, but cannot execute it, so will not change the state of any room until a room returns to state $1$. Whether Charles is currently on line 1 or line 2 of his protocol, he will eventually see the desired state(s) and move to line 3. From here, since $r-1$ rooms are in state $2$, he will eventually flip all of them to $0$, completing the \rpttext block. By this point, he must have visited every room. He will then flip two rooms to states $2$ and $1$, in that order.
		
		At this point, there are rooms in all three states, with exactly one of them in state $1$ (in fact the exact list of states is $0,\ldots,0,1,2$). Alice, Bob and Charles have all visited every room, and will not act again. Whether Deborah is currently on line 1, 2 or 3 she will eventually see the desired states, and proceed to line 4. In particular, she has visited the room currently in state $1$. There are $r-1$ rooms not in state $1$, and she will eventually enter each of them, flipping it to state $1$, and complete the \rpttext block. Her next instruction is to declare; by this point every prisoner has visited every room, and so this declaration is correct.

		The fact that $n\leq 3$ prisoners can win follows from Lemma \ref{monotonic}, although we note that the strategy above straightforwardly adapts to Alice, Bob and Charlie (or Alice and Bob) using the same strategies with \dec appended to the last-named prisoner's protocol. A different winning strategy for $n=2$ was given in \cite{KK21}.
	\end{proof}
	
	\begin{theorem}\label{three-five}The prisoners have a winning strategy whenever $n=5$, $q\geq 3$ and $r\neq 3$.
	\end{theorem}
	\begin{proof}
		The case $r=2$ was proved in \cite{KK21} (and the one-room case is known), so we may assume $r\geq 4$. Again we give a winning strategy assuming all rooms are initially in state $0$; our strategy satisfies the assumptions of \cite[Lemma 4]{KK21} with Alice as leader, and so this implies the existence of a winning strategy for any given starting state.
		
		Alice, Bob and Charlie use the protocols from the proof of Theorem \ref{three-four}. Deborah and Eve both use the same protocol, as follows.
		
		\begin{algorithm}[H]
			\caption{Deborah and Eve's Protocol}
			\see{$(2)$}\;
			\see{$(0)$}\;
			\flip{$(1,0)$}\;
			\eIf{$(1)$}{\see{$(2)$}\;\rpt{$(r-1)$}{\flip{$(0,2)$}}\rpt{$(r-3)$}{\flip{$(2,0)$}}}{\rpt{$(r-1)$}{\flip{$(\bullet,1)$}}\dec}
			\flip{$(2,1)$}\;
			\flip{$(2,1)$}
		\end{algorithm}
		
		As in the proof of Theorem \ref{three-four}, neither Deborah nor Eve will be able to execute line 3 of their protocol (since it implicitly requires seeing $1$) until Alice, Bob and Charlie have all completed their protocols. At this point Alice, Bob and Charlie will each have entered every room, none of them will act again, and there will be $r-2$ rooms in state $0$ and one room in each of the other states. 
		
		Consequently, one of Deborah and Eve will eventually execute line 3 of their protocol, since there are rooms of all states available. Suppose, without loss of generality, that Deborah does this first. Immediately after this action, there are $r-1$ rooms in state $0$ and the remaining room is in state $2$; in particular, Eve will not be able to execute line 3 until Deborah flips a room to state $1$. Now Deborah will eventually execute the whole of the \wotext block without seeing a room in state $1$ in the interim. While doing this she will enter the room in state $2$ and subsequently flip all other rooms into state $2$, thus ensuring she has entered all rooms, then flip $r-3$ rooms to state $0$. When she completes the \wotext block, there are $r-3\geq 1$ rooms in state $0$ and three in state $2$. Since there are still no rooms in state $1$, Eve still cannot execute line 3. Since Deborah has completed the \wotext block, she has reached line 14. Eventually she will enter one of the rooms in state $2$, and execute line 14.
		
		At this point one of two things might happen next. Either Deborah will execute line 15, or Eve will execute line 3. We first assume the former. Now Deborah will not act again, and there are two rooms in state $1$, one in state $2$ and $r-3\geq 1$ in state $0$. Eve will therefore eventually execute lines 1 and 2 (if she hasn't already) and line 3. At this point there are $r-2$ rooms in state $0$, one room in state $1$ and one in state $2$.
		Thereafter, she will not be able to complete the \wotext block (since there are not enough rooms in state $0$) and must eventually enter the room in state $1$, causing her to execute the \eltext block. In order to complete the \eltext block she must enter all $r-1$ rooms not currently in state $1$, and flip them to state $1$; since she previously entered the room which was already in state $1$ this ensures she has entered all rooms by the time she declares.
		
		The other possibility is that Eve executes line 3 when Deborah has executed line 14, but not line 15. After Eve does this, there are $r-2$ rooms in state $0$ and two in state $2$. In this case Eve may execute line 5, but cannot execute line 7 more than $r-2$ times, and so she cannot proceed past line 7 of the \wotext block. Consequently, Eve will not be able to move a room out of state $2$ before Deborah executes line 15. Therefore Deborah will eventually enter one of the remaining rooms in state $2$ and flip it to state $1$. At this point there are $r-2-k$ rooms in state $0$, one in state $1$ and $k+1$ in state $2$, where $k\leq r-2$ is the number of times Eve has executed line 7. Since there are insufficient rooms in state $0$ for Eve to complete lines 6 and 7, she will eventually enter the room in state $1$ and proceed to the \eltext block. Subsequently she will enter all other rooms and declare, as before.
	\end{proof}
	
	\section{Symmetric strategies}\label{sec:symm}
	In this section, we consider for which values of $n$ and $r$ the prisoners have a symmetric winning strategy for some number of states. We may assume all rooms start in state $0$, since any known initial state can be reduced to this at the cost of wasting at most $r-1$ states. We prove the following result.
	\begin{thmtwo}
		If there are $n$ prisoners and $r$ rooms, the prisoners have a symmetric winning strategy for some $q$ if and only if $\hcf(n,r)=1$. Furthermore, there is some function $f$ such that for every pair $n,r$ with $\hcf(n,r)=1$, the prisoners have a winning strategy, starting from all rooms in state $0$, for any $q\geq f(\min(n,r))$.
	\end{thmtwo}
	We break this into three propositions: first we show that the condition $\hcf(n,r)$ is necessary; then we give a strategy using $q=q(n)$ states that works for any coprime $r>n$; and finally a strategy using $q=q(r)$ states that works for any coprime $n>r$ by reducing to the previous case. 
	\begin{proposition}\label{not-coprime}If $\hcf(n,r)>1$ then the prisoners do not have a symmetric winning strategy for any $q$.\end{proposition}
	\begin{proof}We first give a high-level description of the proof before formally specifying how the warden can defeat the prisoners. The idea is to divide the rooms, prisoners and time steps into blocks of size $d=\hcf(n,r)$. In each time block, the warden chooses a block of prisoners and a block of rooms and leads each prisoner from the chosen block into a different room from the chosen block. If the prisoners are using symmetric strategies, at the end of every time block rooms in the same block have the same state, and prisoners in the same block have seen the same sequence of states. The prisoners may have some information about which blocks of prisoners have visited which blocks of rooms, but they cannot deduce that an individual prisoner has visited an individual room.
		
		Let $d=\hcf(n,r)$, and number the prisoners with $[n]$ and the rooms with $[r]$ (here we write $[a]$ to mean $\{1,\ldots,a\}$ for $a\in \mathbb Z^+$). For each $\ell\geq 0$, choose a permutation $\sigma_\ell:[d]\to[d]$. We use these permutations to define a schedule as follows. At time $t$, set $\ell=\lfloor t/d \rfloor$, $k=t-\ell d$ (so $k\equiv t\pmod d$), and let $i\in [n/d]$ such that $i\equiv \ell\pmod{n/d}$ and $j\in [r/d]$ such that $j\equiv \ell\pmod{r/d}$. Now lead prisoner $id+k$ into room $jd+\sigma_{\ell}(k)$.
		
		Note that this is a valid schedule provided that for every $i\in[n/d]$, $j\in[r/d]$ and $k,k'\in[d]$ there are infinitely many $\ell$ such that $\ell\equiv i\pmod{n/d}$, $\ell\equiv j\pmod{r/d}$, and $\sigma_\ell(k)=k'$.
		
		Fix any symmetric strategy $\mathcal S$ for the prisoners. We claim that for each $t\equiv 0\pmod d$ and each $i\in [n/d]$ and $j\in[r/d]$ there are states $s(t,j,\mathcal S)$ and lists of states $L(t,i,\mathcal S)$, which do not depend on the choice of $\sigma_1,\sigma_2,\ldots$, such that after $t$ steps, for each $k\in [d]$, room $dj+k$ is in state $s(t,j,\mathcal S)$ and the list of states observed by prisoner $di+k$ is $L(t,i,\mathcal S)$. This follows by induction. It is true for $t=0$ with each state being $0$ and each list being empty. For $t>0$ we observe that between times $t-d$ and $t$, by the induction hypothesis, $d$ prisoners who have observed the same list of previous states are led into $d$ rooms in the same state; since they follow the same strategy, each changes that room to the same new state, irrespective of which prisoner is led into which room.
		
		Consequently the prisoners never obtain any information about the sequence $\sigma_1,\sigma_2,\ldots$, and since the first time by which all prisoners have been led into all rooms depends on these and can be arbitrarily large, they can never know that it has occurred.\end{proof}
	
	To show that the prisoners do have a symmetric winning strategy (for sufficiently large $q$) in all other cases, we first note that this essentially reduces to whether they can find a leader. 
	\begin{lemma}\label{leader-wins}Suppose that, for $n$ prisoners, $r$ rooms and $q$ states, the prisoners have a symmetric strategy ensuring that exactly one prisoner will at some point consider themselves a leader. Then $n$ prisoners have a symmetric winning strategy for $r$ rooms with $q+4$ states.\end{lemma}
	\begin{proof}The prisoners first follow the hypothesised strategy using states $0,\ldots,q-1$. One prisoner eventually becomes the leader. The leader then flips every room they see in a state less than $q$ to state $q$, until this has been done $r$ times. 
		
		Recall that the prisoners have a winning strategy using $4$ states \cite[Theorem 5]{KK21}. Furthermore, this strategy involves one leader, with all other prisoners following the same protocol. We now use this strategy on states $\{q,q+1,q+2,q+3\}$: any prisoner who sees any of these states without having become the leader switches to the non-leader's protocol using those four states (ignoring any other states they see in the interim), whereas the leader, after moving $r$ rooms to state $q$, follows the leader's protocol on those four states.
	\end{proof}
	\begin{proposition}\label{n-smaller}
		For every $n$ there exists $q=q(n)$ such that for every $r>n$ with $\hcf(n,r)=1$, $n$ prisoners have a symmetric winning strategy for $r$ rooms with $q$ states.
	\end{proposition}
	\begin{proof}
		By Lemma \ref{leader-wins}, it is sufficient to show that the prisoners can find a leader using $q-4$ states. For simplicity of argument, we describe the leader-finding strategy using two switches in each room: a binary switch $S$ and an $n$-state switch $T$. These can be mapped onto states $0,\ldots, 2n-1$, giving $q=2n+4$.
		No prisoner will ever decrease the state of any $S$- or $T$-switch. When a prisoner first sees a room with the $S$-switch set to $0$, they move it to a $1$. They do this regardless of everything else. After this they never touch an $S$-switch again 
		Since $n < r$ this means that eventually $n$ rooms will be marked.
		
		We now use the following ``race'' strategy to find a leader. Each prisoner starts at level $n$. 
		On level $i$, each prisoner tries to move the $T$-switch in a certain number of rooms from any state that is at most $n-i$ to state $n-i+1$. 
		If $r/i$ is not an integer, prisoners on level $i$ attempt to do this in $\lceil r/i\rceil$ rooms. If $r/i$ is an integer, they attempt to do this in $\lceil n/i\rceil$
		rooms that have their $S$-switch set to $1$. Any prisoner who succeeds moves to level $i-1$.
		
		We argue by induction that at least one and at most $i$ prisoners reach level $i$. This is trivial
		for level $n$. Suppose this holds for $i > 1$. Since $\hcf(n, r) = 1$, one of $r/i$ and $n/i$ is not an integer,
		so either the target is $\lceil r/i\rceil$ and $i\lceil r/i\rceil>r$ or the target is $\lceil n/i\rceil$ marked rooms and $i\lceil n/i\rceil>n$; 
		in either case it is impossible for $i$ prisoners to reach the target. Suppose that no prisoners reach level $i-1$. By the induction
		hypothesis, at most $i$ prisoners (and at least one) reached level $i$. If the target is $\lceil r/i\rceil$, at least one of these prisoners must reach it,
		since otherwise the $T$-switches in at most $i(\lceil r/i\rceil-1) < r$ rooms reach state $n-i+1$ or above, and eventually one of the prisoners will be able to change another one.
		Similarly, if the target is $\lceil n/i\rceil$, eventually there will be $n$ marked rooms and thereafter at least one prisoner must reach the target from among these rooms.
		
		The unique prisoner who reaches level $1$ becomes the leader. 
	\end{proof}
	
	\begin{proposition}\label{r-smaller}
		For every $r$ there exists $q=q(r)$ such that for every $n>r$ with $\hcf(n,r)=1$, $n$ prisoners have a symmetric winning strategy for $r$ rooms with $q$ states.
	\end{proposition}
	\begin{proof}This is known for $r=1$ (see \cite{Win04}), so we may assume $r\geq 2$.
		We give a protocol that ensures exactly $r-1$ prisoners eventually become candidate leaders. From here, we can use the strategy of Proposition \ref{n-smaller} to identify a single leader from the $r-1<r$ candidates. 
		
		We find candidate leaders using two binary switches labelled $S$ and $T$. The remaining process of finding a leader from the candidates can be achieved with an additional $2(r-1)$-state switch, for a total of $8r-8$ states, and then Lemma \ref{leader-wins} gives $q(r)=8r-4$.
		
		Since $\hcf(n,r)=1$, there exist positive integers $k,m$ such that $kn=mr-1$. Each prisoner does three things. 
		
		First, if they enter a room with the $S$-switch in position $0$ they switch it a $1$ and designate themselves Type I. This does \emph{not} need to be the first room they enter, but they only do this once. This means that eventually there will be exactly $r$ Type I prisoners.
		
		Secondly, each prisoner moves a $T$-switch from $0$ to $1$ on their first $k$ opportunities. They do this regardless of whether they are Type I, Type II (defined below), or neither.
		
		Thirdly, if they are of Type I then they attempt to move a $T$-switch from $1$ to $0$ exactly $\lceil \frac{kn}{r}\rceil=m$ times. If they succeed they designate themselves Type II. We permit (although this is not necessary) multiple moves in a single visit to a room, so a prisoner needing to move a $T$-switch in both directions will do this until they have completed at least one direction before leaving the room.
		
		We claim that exactly $r-1$ prisoners will designate themselves Type II. To see this, observe that every $T$-switch in state $1$ will eventually be moved to a $0$. This is because $T$-switches will be moved to state $1$ only $kn$ times in total and there are $r$ Type I prisoners each of whom is trying to move a $T$-switch from $1$ to $0$ on $m$ occasions, and $mr=kn+1$. Moreover, since none of the $r$ Type I prisoners does this more than $m$ times, exactly $r-1$ of them must do it exactly $m$ times, with the remaining prisoner only achieving $m-1$ flips.
		
		The $r-1$ Type II prisoners are the candidate leaders required.
	\end{proof}
	
	\begin{proof}[Proof of Theorem \ref{thm:coprime}]
		This is immediate from Propositions \ref{not-coprime}, \ref{n-smaller} and \ref{r-smaller}.
	\end{proof}
	
	\section{Open questions}
	The results of Section \ref{sec:three} show that up to five prisoners can win with three states, except possibly if there are exactly three rooms. 
	This raises the question of what happens with more prisoners.
	\begin{question}\label{five-three}Can five prisoners win if there are exactly three rooms and three states? If so, is there some value $n>5$ such that $n$ prisoners cannot win?
	\end{question}
	Note that a negative answer to the first part of Question \ref{five-three} would be particularly interesting, since it would show that there is no monotonicity in the number of rooms.
	\begin{question}Can more prisoners win if there are sufficiently many rooms? In particular, does there exist some function $f(n)$ such that $n$ prisoners can win in $r$ rooms with $3$ states provided $r\geq f(n)$?
	\end{question}
	Turning to the case of symmetric strategies, it is perhaps surprising that, in cases where the prisoners win for a sufficiently large number of states, that the number of states required (if all rooms are initially in the same state) does not increase unless \textit{both} $n$ and $r$ are large. However, we have not shown that arbitrarily large $q$ can be necessary even in that case.
	\begin{question}
		Is there some absolute constant $q_0$ such that whenever $\hcf(n,r)=1$ there is a symmetric winning strategy for $n$ prisoners in $r$ rooms with $q_0$ states, starting with all rooms in state $0$?
	\end{question}
	
	\section*{Acknowledgements}
	Part of this work was carried out while J.H. was supported by by the European Research Council under the European Union’s Horizon 2020 research and innovation programme (grant agreement no.\ 883810). For the purpose of open access, the authors have applied a Creative Commons Attribution (CC BY) licence to any Author Accepted Manuscript version arising from this work.

\end{document}